\documentclass[12pt] {article}
\usepackage{latexsym}
\usepackage[dvips]{graphics,color}
\usepackage{graphpap}
\usepackage{alltt}
\usepackage{graphicx}
\usepackage{amsfonts}
\usepackage{amsmath}
\usepackage{amsthm}
\usepackage{psfrag}
\usepackage{color}
\newtheorem{theorem}{Theorem}[section]

\newtheorem{proposition}[theorem]{Proposition}
\newtheorem{definition}[theorem]{Definition}
\newcommand{\beq}{\begin{equation}}
\newcommand{\feq}[1]{\label{#1} \end{equation}}
\newcommand{\beqr}{\begin{eqnarray}}
\newcommand{\feqr}{\end{eqnarray}}
\def\non{\nonumber}

\newcommand{\rf}[1]{(\ref{#1})}

\def\pre#1#2#3{Phys. Rev. {\bf{E#1}} (#2) #3}

\def\rmp#1#2#3{Rev. Mod. Phys. {\bf{#1}} (#2) #3}

\def\ap#1#2#3{Ann. of Phys. {\bf{#1}} (#2) #3}
\def\apk#1#2#3{Ann. Physik {\bf{#1}} (#2) #3}

\def\chaos#1#2#3{Chaos {\bf{#1}} (#2) #3}

\def\jmp#1#2#3{J. Math. Phys. {\bf{#1}} (#2) #3}

\def\rmp#1#2#3{Rev. Mod. Phys. {\bf{A#1}} (#2) #3}

\def\laa#1#2#3{Linear Algebra Appl. {\bf{#1}} (#2) #3}

\begin{document}

\begin{center}


{\Large \bf The fractional Schr\"{o}dinger operator and Toeplitz matrices.}\\
[4mm]

\large{Agapitos Hatzinikitas} \\ [5mm]

{\small Department of Mathematics, \\
School of Sciences, \\
University of Aegean, \\
Karlovasi, 83200\\
Samos, Greece \\
Email: ahatz@aegean.gr}\\ [5mm]

\end{center}
\begin{abstract}
Confining a quantum particle in a compact subinterval of the real line with Dirichlet boundary conditions, we identify the connection of the one-dimensional fractional Schr\"{o}dinger operator with the  truncated Toeplitz matrices. We determine the asymptotic behaviour of the product of eigenvalues for the $\alpha$-stable symmetric laws by employing the Szeg\"{o}'s strong limit theorem. The results of the present work can be applied to a recently proposed model for a particle hopping on a bounded interval in one dimension whose hopping probability is given a discrete representation of the fractional Laplacian. 
\end{abstract}

\section{Introduction}
\label{sec0}
\par It is well known that a stable law \cite{Ref1,Ref2,Ref3,Ref3a,Ref3b}, which is a direct generalization of the Gaussian distribution, is generated by the parameters $(\alpha, \beta, c, \tau)$ where: 
\begin{description}
\item[$-$] $\alpha$ is the characteristic exponent that determines the degree of leptokurtosis
and the fatness of the tails. In the present work we consider $\alpha \in (1,2)$.
\item[$-$] $\beta$ is the skewness parameter which characterizes the degree of asymmetry
of the L$\acute{e}$vy measure and takes values in the interval $[-1,1]$.  
\item[$-$] c is the scale parameter with range $(0,\infty)$ and measures
scale in place of standard deviation.
\item[$-$] $\tau$ is the location parameter which saturates the set of real numbers
and shifts the distribution to the left or right.
\end{description}
It is often denoted by $S_{\alpha}(\beta,c,\tau)$. We adopt this notation to declare different parametric families of operators one might consider at quantum level. Our study will be focused on $S_{\alpha}(c)$ ($\alpha$-stable symmetric) operators.
\par In the literature the papers \cite{Ref4,Ref5} claim that the time-independent fractional Schr\"{o}dinger equation with infinite one-dimensional square well potential has the same eigenfunctions as for the standard non-fractional case, only with modified energies. Recently in \cite{Ref6}, it was argued via a proof by contradiction, that the ground state cannot be a solution either in the interval $\alpha \in (-1,1)$ (unless $\alpha=0$) or in $(1,2)$. Nevertheless, by using the Gr\"{u}nwald-Letnikov definition for the fractional derivative, one can numerically evaluate the eigenvalues and extract information about the asymptotic behaviour of the product of eigenvalues without knowing their explicit form. To prove this statement we organize our paper as follows.  
\par In Section II the square well potential with perfectly rigid walls serves as a simpified model from which one can derive and compare quantities (such as the determinant and trace) stemming from analytical and discretization methods. At discretized level one encounters the real symmetric Toeplitz matrix \rf{sec01 : eq7} which forms a subset of the class of symmetric centrosymmetric matrices \cite{Ref11}. The eigenvalue problem can be solved exactly \rf{sec01 : eq71a} and the asymptotic behaviour of the product of eigenvalues is an easily accessible task since a closed expression for the determinant of this matrix can be determined. 
\par In Section III the situation for time-independent fractional operators in the class $S_{\alpha}(\beta,c)$, is more involved. Adopting the Gr\"{u}nwald-Letnikov definition for fractional derivatives with anisotropic coefficients, retaining Dirichlet boundary conditions and using the discretization method, we end up with the finite dimensional matrix \rf{sec1 : eq11}. This matrix is recognized to be a Toeplitz matrix with well known properties. Restricting to the subclass of $S_{\alpha}(c)$ laws one can prove that its eigenvalues are simple and non-negative thus positive definite. 
\par In Section IV the symbol (or generating function) $f$ of the Toeplitz matrix \rf{sec1 : eq11} is shown to belong to the intersection of the Wiener and Besov spaces, and all requirements of Sz\"{e}go's strong limit theorem are satisfied except that the zeros of $f$ are located at even integer multiples of $\pi$. Excluding these points from the complex unit circle the logarithm of the symbol can be expanded in Fourier modes while the determinant of \rf{sec1 : eq11} diverges as one might expect since the point spectrum of the associated operator is unbounded from above. Our final result for the asymptotic behaviour of the product of eigenvalues is captured by \rf{sec2 : eq6} and \rf{sec3 : eq16} which hold for every rational $\alpha$ in the superdiffusion region $(1,2)$.  
\section{The square well potential as a toy model}
\label{sec01}

We study the asymptotic behaviour of the spectrum for a particle confined in an one-dimensional square well potential of length $L$ with perfectly rigid walls. The potential is defined as 
\begin{eqnarray}
V(x)=\left\{ \begin{array}{ll} 0, & 0<x<L \\ \infty, & x\leq 0, x\geq L. \end{array} \right. 
\label{sec01 : eq1}
\end{eqnarray}
The usual time-idependent Schr\"{o}dinger equation with Dirichlet boundary conditions is 
\begin{eqnarray}
-\frac{\hbar^2}{2m}\frac{d^2}{dx^2}\psi(x)&=&E\psi(x), \quad 0<x<L \non \\
\psi(x)&=& 0, \quad \forall x\notin(0,L) 
\label{sec01 : eq2}
\end{eqnarray}
and has solutions the normalized eigenfunctions, 
\begin{eqnarray}
\psi_n(x)=\left\{\begin{array}{ll} \sqrt{\frac{2}{L}}\sin\left(\frac{n\pi}{L}x\right), & n\in {\mathbb N}, \,\, 0<x<L  \\ 0, & \quad \forall x\notin(0,L). \end{array}\right.
\label{sec01 : eq3}
\end{eqnarray}
The operator $-\frac{d^2}{dx^2}$ has a discrete spectrum unbounded from above with points 
\begin{eqnarray}
E_n=\frac{\hbar^2}{2m}\left(\frac{n\pi}{L}\right)^{2}, \,\, n\in {\mathbb N}.
\label{sec01 : eq4}
\end{eqnarray}
The product of $N$ eigenvalues is
\begin{eqnarray}
\prod_{n=1}^N E_n=\left(\frac{\hbar^2 \pi^2}{2mL^2} \right)^N (N!)^2 
\approx \frac{2}{e}\left(\frac{\hbar^2}{2mL^2} \right)^N \left(\frac{\pi}{e}\right)^{2N+1}(N+1)^{2N+1}
\label{sec01 : eq5} 
\end{eqnarray}
where the Stirling's formula $N!\approx \sqrt{2\pi} (N+1)^{N+\frac{1}{2}}e^{-(N+1)}$ for large $N$ has been applied. The sum of the $N$ eigenvalues is
\begin{eqnarray}
\sum_{n=1}^N E_n=\left(\frac{\hbar^2 \pi^2}{2mL^2} \right)\sum_{n=1}^N n^2=\left(\frac{\hbar^2 \pi^2}{2mL^2} \right)\left(\frac{1}{3}N^2+\frac{1}{2}N+\frac{1}{6}\right)N.
\label{sec01 : eq51} 
\end{eqnarray}
\par In the discretized method we consider a grid of $N'-$ordered points of the interval $[0,L]$
\begin{eqnarray}
x_0=0<x_1<\cdots<x_{N'}<x_{N'+1}=L
\label{sec01 : eq52} 
\end{eqnarray}
with equal spacing $x_{i+1}-x_i=\epsilon>0, \,\, \forall i=0,\cdots,N'$. Then by employing the centered second difference estimator of the second ordered derivative, namely
\begin{eqnarray}
\frac{d^2}{dx^2}\psi(x)=\lim_{\epsilon\rightarrow 0^+}\frac{\psi(x+\epsilon)-2\psi(x)+\psi(x-\epsilon)}{\epsilon^2}
\label{sec01 : eq6}
\end{eqnarray}
and denoting the values of the function $\psi$ at the grid points by $\psi_i=\psi(x_i)$, 
the problem \rf{sec01 : eq2} is equivalent to 
\begin{eqnarray}
&& \left(\mathcal{A}_{N'}\right)_{ij} \psi_j=E_i \psi_i, j=1,\cdots,N' \non \\
&\rm{with}& \psi_{-k}=\psi_0=0=\psi_{N'+k}, \quad \forall k\in {\mathbb N}
\label{sec01 : eq53} 
\end{eqnarray}
where
\begin{eqnarray}
\mathcal{A}_{N'}=\frac{\hbar^2}{2m\epsilon^2}\mathcal{B}_{N'}=\frac{\hbar^2}{2m\epsilon^2}\left(\begin{array}{rrrrrr} 2 & -1 & 0 & \cdots & 0 & 0 \\
-1 & 2 & -1 & \cdots & 0 & 0 \\
\vdots & \vdots & \vdots & \cdots & \vdots & \vdots \\
0 & 0 & 0 & \cdots & -1 & 2 \\  
\end{array}\right)_{N'\times N'}.
\label{sec01 : eq7} 
\end{eqnarray}
The eigenvalues and eigenvectors of $\mathcal{B}_{N'}$ are \cite{Ref7}
\begin{eqnarray}
\lambda_k &=& 2\left(1+\cos\left(\frac{k\pi}{N'+1}\right)\right), \, 1\leq k \leq N' \non \\
e_{mk} &=& \sin\left(\frac{k(m+1)\pi}{N'+1}\right), \, 0\leq m \leq N'-1.
\label{sec01 : eq71a} 
\end{eqnarray}
The trace of $\mathcal{A}_{N'}$ is
\begin{eqnarray}
\rm{Tr}(\mathcal{A}_{N'})&=&\frac{\hbar^2}{2m\epsilon^2}\rm{Tr}(\mathcal{B}_{N'})=\frac{\hbar^2}{2mL^2}(N'+1)^2\left(2N'+\sum_{k=1}^{N'}\cos\left(\frac{k\pi}{N'+1}\right)\right) \non \\
&=&\frac{\hbar^2}{2mL^2}(N'+1)^2 \,2N'
\label{sec01 : eq71}
\end{eqnarray}
while its determinant is given recursively by
\begin{eqnarray}
\det \mathcal{A}_{N'}=2 \left(\frac{\hbar^2}{2m \epsilon^2}\right) \det \mathcal{A}_{N'-1}-\left(\frac{\hbar^2}{2m \epsilon^2}\right)^2 \det \mathcal{A}_{N'-2}
\label{sec01 : eq8} 
\end{eqnarray}
with solution
\begin{eqnarray}
\det \mathcal{A}_{N'}=\left(\frac{\hbar^2}{2m L^2}\right)^{N'} (N'+1)^{2N'+1}.
\label{sec01 : eq 9}
\end{eqnarray}
The determinant could had equally been evaluated using the eigenvalues \rf{sec01 : eq71a} 
\begin{eqnarray}
\det \mathcal{B}_{N'}=\prod_{k=1}^{N'}\lambda_k=\prod_{k=1}^{N'}(2+2\cos\left(\frac{k\pi}{N'+1}\right))=N'+1.
\label{sec01 : eq91} 
\end{eqnarray}
Increasing the number $N'$ of intermediate points we approach the exact eigenvalues of problem \rf{sec01 : eq2} with higher accuracy. 
\section{The discretized fractional Schr\"{o}dinger operator versus Toeplitz matrices}
\label{sec1}

We consider the fractional operator \cite{Ref8}
\begin{eqnarray}
\mathcal{\hat{A}}(\alpha,\beta)=\frac{1}{\cos \left(\frac{\pi\alpha}{2}\right)}\left(K^+\frac{d^{\alpha}}{dx^{\alpha}}+K^-\frac{d^{\alpha}}{d(-x)^{\alpha}}\right) 
\label{sec1 : eq1}
\end{eqnarray}
where $K^{\pm}$ are anisoptropic constants satisfying $K^{\pm}\geq 0$ and $K^++K^->0$. Operator \rf{sec1 : eq1} can be casted into the equivalent form 
\begin{eqnarray}
\mathcal{\hat{A}}(\alpha,\beta)=\frac{K_{\alpha}}{2\cos \left(\frac{\pi\alpha}{2}\right)}\left((1+\beta)\frac{d^{\alpha}}{dx^{\alpha}}+(1-\beta)\frac{d^{\alpha}}{d(-x)^{\alpha}}\right) 
\label{sec1 : eq2}
\end{eqnarray}
where $[K_{\alpha}]=\frac{[M][L]^{\alpha+2}}{[T]^2}$ and 
\begin{eqnarray}
K_{\alpha}=K^++K^-, \quad \alpha\in (0,2)/\{1\}, \quad \beta=\frac{K^+-K^-}{K^++K^-}, \quad \beta\in [-1,1].
\label{sec1 : eq3} 
\end{eqnarray}
For the case $\alpha=1$ one should consider an operator as the one proposed in \cite{Ref8}. We study only the superdiffusive region $1< \alpha <2$ since the subdiffusive can be treated on equal footing with minor modifications. 
\begin{definition}
The derivatives in \rf{sec1 : eq2} for  $L^1_{(0,L)}$ functions are defined by 
\begin{eqnarray}
\left(\frac{d^{\alpha}}{dx^{\alpha}}\psi\right)(x)&=&\lim_{\epsilon\rightarrow 0^+}\frac{\left(\Delta^{\alpha}_{\epsilon}\psi\right)(x)}{\epsilon^{\alpha}}=\lim_{\epsilon\rightarrow 0^+}\frac{1}{\epsilon^{\alpha}}\sum_{n=0}^{\left[\frac{x}{\epsilon}\right]}w_n(\alpha) \psi(x-(n-1)\epsilon), \,\, \rm{for} \,\, x>0 \non \\
\left(\frac{d^{\alpha}}{d(-x)^{\alpha}}\psi\right)(x)&\!\!=&\!\! \lim_{\epsilon\rightarrow 0^+}\frac{\left(\Delta^{\alpha}_{-\epsilon}\psi\right)(x)}{\epsilon^{\alpha}}=\lim_{\epsilon\rightarrow 0^+}\frac{1}{\epsilon^{\alpha}}\sum_{n=0}^{\left[\frac{L-x}{\epsilon}\right]}w_n(\alpha) \psi(x+(n-1)\epsilon), \,\, \rm{for} \,\, x<L 
\label{sec1 : eq4}
\end{eqnarray}
and are called the right-handed (left-handed) Gr\"{u}nwald-Letnikov fractional derivatives \cite{Ref9} \footnote{Actually this is a variant of the Gr\"{u}nwald-Letnikov fractional derivative, in which the function evaluations are shifted to the right.}. The weights $w_n(\alpha)$ are given recursively by \footnote{Another way to compute the coefficients $w_n(\alpha)$ is by their generating function $w(z)=(1-z)^{\alpha}$ which has the Taylor expansion $w(z)=\sum_{k=0}^{\infty}(-1)^k \left(\begin{array}{l} \alpha \\ k \end{array} \right) z^k$ for $|z|\leq 1$ and $\alpha\in [1,2]$. }
\begin{eqnarray}
w_0(\alpha)&=&1, \non \\
\quad w_n(\alpha)&=& \left(1-\frac{\alpha+1}{n}\right)w_{n-1}(\alpha)=(-1)^n \left(\begin{array}{l} \alpha \\ n \end{array} \right)=\frac{\alpha(\alpha-1)\Gamma(n-\alpha)}{\Gamma(2-\alpha)\Gamma(n+1)}, \non \\
&& \forall n\in {\mathbb N}.
\label{sec1 : eq5}
\end{eqnarray}
\end{definition}
When $\alpha=[\alpha]$, a positive integer, then $w_n([\alpha])=0, \,\, \forall n\geq [\alpha]+1$. The only negative weight in the superdiffusive region is $w_1(\alpha)$ \footnote{In the subdiffusive region all weights are negative except $w_0$.} and the infinite series of weights converges to zero (see \rf{apAb1} for the proof)
\begin{eqnarray}
\sum_{n=0}^{\infty}w_n(\alpha)=0.
\label{sec1 : eq6}
\end{eqnarray}
Substituting $\alpha=2$ and  $\beta=0$ into \rf{sec1 : eq4} we recover the centered second difference estimator \rf{sec01 : eq6} since $w_n(2)=0, \,\, \forall n\geq 3$.
\par The eigenvalue problem 
\begin{eqnarray}
&& \left(\mathcal{\hat{A}}(\alpha,\beta)\psi \right)(x)=\lambda \psi(x), \quad \psi \in L^2_{(0,L)} \non \\
&\rm{with}& \psi(x)=0, \quad \forall x\notin (0,L) 
\label{sec1 : eq8} 
\end{eqnarray}
using a grid of n-ordered points, as in Section II, gives the discretized problem
\begin{eqnarray}
&& \left(\mathcal{A}_{n}(\alpha,\beta)\right)_{ij} \psi_j=\lambda \psi_i, j=1,\cdots,n \non \\
&\rm{with}& \psi_{-k}=\psi_0=0=\psi_{n+k}, \quad \forall k\in {\mathbb N}.
\label{sec1 : eq10} 
\end{eqnarray}
The approximation of the operator $\hat{\mathcal{A}}(\alpha,\beta)$ by the sum of two matrices corresponding to the left- and right-fractional derivatives, results in a ${\mathbb R}^{n\times n}$ matrix with elements 
\begin{eqnarray}
\left(\mathcal{A}_{n}(\alpha,\beta)\right)_{ij}=\frac{K_{\alpha}}{2\cos \left(\frac{\pi \alpha}{2}\right)\epsilon^{\alpha}}\times \left\{\begin{array}{llll} 2w_1 & \rm{for} & j=i\,;\, & i=1,\cdots , n \\ 
(1+\beta)w_0+(1-\beta) w_2 & \rm{for} & j=i+1 \, ; \, & i=1,\cdots , n-1 \\ 
(1-\beta) w_{3} & \rm{for} & j=i+2 \, ; \, & i=1,\cdots, n-2 \\ 
(1-\beta) w_{4} & \rm{for} & j=i+3 \, ; \, & i=1,\cdots, n-3 \\ 
\vdots & \vdots & \vdots & \vdots \\ 
(1-\beta) w_{n-1} & \rm{for} & j=i+n-2 \, ; \, & i=1,2 \\ 
(1-\beta) w_{n} & \rm{for} & j=i+n-1 \, ; \,& i=1. \end{array}\right. \non \\
\label{sec1 : eq11}
\end{eqnarray}
Note that we do not allow the end points $x_0=0$ and $x_{n+1}=L$ to be present in the calculation. The elements of the matrix parallel to the main diagonal are equal and \rf{sec1 : eq11} is recognized to be a Toeplitz matrix. Renaming the entries of the matrix according to
\begin{eqnarray}
\left(\mathcal{A}_{n}(\alpha,\beta)\right)_{ij}=\frac{K_{\alpha}}{2\cos \left(\frac{\pi \alpha}{2}\right)\epsilon^{\alpha}}(\tilde{a}_{i-j}(\alpha,\beta)), \quad i,j=1,\cdots, n
\label{sec1 : eq15}
\end{eqnarray}  
its symbol in general is given by 
\begin{eqnarray}
f(\alpha,\beta;z;n)=\sum_{k=-n}^{n}f_k(\alpha,\beta) z^k, \quad z\in{\mathbb C}.
\label{sec1 : eq16}
\end{eqnarray}
This is a rational function of z which is analytic everywhere on the complex plane except at the $n$-poles. On the counterclockwise complex unit circle ${\mathbb T}=\{z\in {\mathbb C}: |z|=1\}$ the generating function \rf{sec1 : eq16}, for a large ordered sequence $\{x_n\}_{n\in {\mathbb N}}$ of intermediate points of the interval $(0,L)$, is given by the expression
\begin{eqnarray}
\lim_{n\rightarrow \infty}f(\alpha,\beta;{\mathbb T};n)&=&f(\alpha,\beta;e^{i\theta};\infty)=u(\alpha;\theta;\infty)+iv(\alpha,\beta;\theta;\infty), \quad \textrm{where} \label{sec1 : eq17} \\
u(\alpha;\theta;\infty)&=& \frac{K_{\alpha}}{\left|\cos \left(\frac{\pi \alpha}{2}\right)\right|\epsilon^{\alpha}}\left[\alpha-\cos \theta -\frac{\alpha(\alpha-1)}{\Gamma(2-\alpha)}\sum_{k=1}^{\infty}\cos(k\theta)\frac{\Gamma(k+1-\alpha)}{\Gamma(k+2)}\right] \non \\
&=&\frac{K_{\alpha}}{\left|\cos \left(\frac{\pi \alpha}{2}\right)\right|\epsilon^{\alpha}} \left|2\sin \left(\frac{\theta}{2}\right) \right|^{\alpha}\cos\left[(\theta-\pi) \left(1-\frac{\alpha}{2}\right)\right] \label{sec1 : eq17a1} \\ 
v(\alpha,\beta;\theta;\infty)&=& \frac{K_{\alpha}\beta}{\left|\cos \left(\frac{\pi \alpha}{2}\right)\right|\epsilon^{\alpha}}\left[\sin \theta-\frac{\alpha(\alpha-1)}{\Gamma(2-\alpha)}\sum_{k=1}^{\infty}\sin(k\theta)\frac{\Gamma(k+1-\alpha)}{\Gamma(k+2)}\right] \non\\
&=&-\frac{K_{\alpha}\beta}{\left|\cos \left(\frac{\pi \alpha}{2}\right)\right|\epsilon^{\alpha}} \left|2\sin \left(\frac{\theta}{2}\right) \right|^{\alpha}\sin\left[(\theta-\pi) \left(1-\frac{\alpha}{2}\right)\right].
\label{sec1 : eq17b1}
\end{eqnarray}
The infinite cosine series is non-negative definite for $\theta\neq 2m\pi$, $2\pi$-periodic and conveges to \rf{sec1 : eq17a1} (see \rf{apA1} for the proof) for $mT<\theta<2\pi+mT$, where $T$ is its period. Function \rf{sec1 : eq17a1} is periodic only for rational values of $\alpha$ since the sine and cosine functions have periods $T_1=2\pi n_1$, $T_2=\frac{4\pi n_2}{2-\alpha}$ and therefore the product function is periodic only for $n_2=n_1\left(1-\frac{\alpha}{2}\right)$. For example the value $\alpha=\frac{3}{2}$ gives $(n_1,n_2)=(4,1)$ and $T=8\pi$ as it is also justified in the following plot.  
\begin{figure}[hp]
\centering
\includegraphics[scale=0.7]{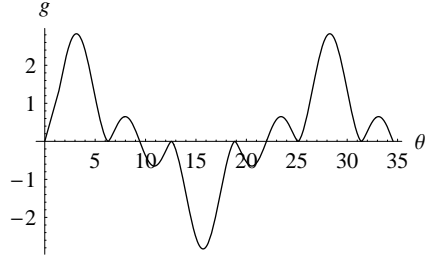}
\caption{The plot of the function $g(\theta)=\left|2\sin \left(\frac{\theta}{2}\right) \right|^{\alpha}\cos\left[(\theta-\pi) \left(1-\frac{\alpha}{2}\right)\right]$ for $\theta\in[0,11\pi]$.}
\end{figure}

We can shift the origin by substituting $\theta\rightarrow \theta-\pi$ in \rf{sec1 : eq17a1} and \rf{sec1 : eq17b1} thus obtaining 
\begin{eqnarray}
u(\alpha;\theta;\infty)&=&\frac{K_{\alpha}}{\left|\cos \left(\frac{\pi \alpha}{2}\right)\right|\epsilon^{\alpha}} \left|2\cos \left(\frac{\theta}{2}\right) \right|^{\alpha}\cos\left[\theta \left(1-\frac{\alpha}{2}\right)\right] \label{sec1 : eq17c1} \\
v(\alpha,\beta;\theta;\infty)&=&-\frac{K_{\alpha}\beta}{\left|\cos \left(\frac{\pi \alpha}{2}\right)\right|\epsilon^{\alpha}} \left|2\cos\left(\frac{\theta}{2}\right) \right|^{\alpha}\sin\left[\theta \left(1-\frac{\alpha}{2}\right)\right]
\label{sec1 : eq17d1}
\end{eqnarray}
where $-\pi+mT<\theta<\pi+mT$. Note that for $\alpha=2$ and $\beta=0$ the symbol becomes
\begin{eqnarray}
f(2,\beta;{\mathbb T};n)=u(2,\beta;{\mathbb T};n)= 2K_{\alpha}(1-\cos \theta)=K_{\alpha}\left(2\sin \frac{\theta}{2}\right)^2 
\label{sec1 : eq18}
\end{eqnarray}
with range $[0,4K_{\alpha}]$.  
\par In the complex plane the symbol is graphically represented by a simple closed contour $\mathcal{C}$.
\begin{figure}[hp]
\centering
\includegraphics[scale=0.7]{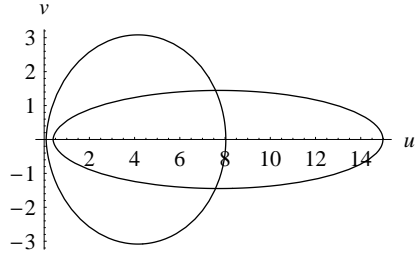}
\caption{Symbol curves $f(\alpha,\beta;{\mathbb T};n)$ in the complex plane for the generating function \rf{sec1 : eq17} with parameters $\alpha=1.5$, $\beta=0.8$ and $\alpha=1.2$, $\beta=0.2$ (the far right contour). The number of intermediate points is $n=200$.}
\end{figure}

For a given value of $\theta \in (0,2\pi)$ the point $(u,v)\in \mathcal{C}$ satisfies the equation 
\begin{eqnarray}
\left(\frac{\cos \left(\frac{\pi \alpha}{2}\right)}{K_{\alpha}}u \right)^2 +\left(\frac{\cos \left(\frac{\pi \alpha}{2}\right)}{\beta K_{\alpha}}v\right)^2 
&=& \left(2\sin \frac{\theta}{2} \right)^{2\alpha}.
\label{sec1 : eq19}
\end{eqnarray}
The matrix \rf{sec1 : eq11} contains the following subclasses of matrices:
\begin{enumerate}
\item[($i$)] The matrix $\mathcal{A}_{n}(\alpha,\beta=0)$ is real symmetric and thus having real eigenvalues. Moreover its eigenvalues are  simple in accordance with the following theorem (see Appendix B for the terminology): 
\begin{theorem}[Trench 1993]
Suppose that $f$ is nonincreasing and 
\begin{eqnarray}
f(0+)=M>m=f(\pi-). 
\end{eqnarray}
Then for every $n$ the matrix $\mathcal{A}_n$ has $n$ distinct eigenvalues in $(m,M)$, its even and odd spectra are interlaced, and its largest eigenvalue is even. 
\end{theorem}
In our case the symbol of the truncated Toeplitz matrix is a decreasing function of $\theta \in (-\pi,0)$ and $f(-\pi+)>f(0-)$ thus the theorem is applicable. 
\begin{proposition}
The matrix $\mathcal{A}_n$ is strictly positive definite. 
\end{proposition}
This is a consequence of the previous theorem since all the eigenvalues belong to the interval $(f(0-),f(-\pi+))$ with $f(0-)>0$.
\par Diagonalizing $\mathcal{A}_n$ by an orthogonal matrix and using the cyclicity property of the trace we obtain
\begin{eqnarray}
\textrm{Tr}(\mathcal{A}_n)=-\frac{n w_1 K_{\alpha}}{\left|\cos \left(\frac{\pi \alpha}{2}\right)\right|\epsilon^{\alpha}}=\frac{K_{\alpha}}{\left|\cos \left(\frac{\pi \alpha}{2}\right)\right|\epsilon^{\alpha}}\sum_{i=1}^n \lambda_i\Rightarrow -w_1(\alpha)=\frac{1}{n}\sum_{i=1}^n \lambda_i>0,
\label{sec1 : eq12}
\end{eqnarray} 
where $\lambda_i$ are the eigenvalues of $\mathcal{A}_n$. Also combining \rf{sec1 : eq12} and \rf{sec1 : eq5} we can express $w_q$ in terms of $\lambda_i$'s 
\begin{eqnarray}
w_q(\alpha)=\frac{(1-\alpha)_{q-1}}{\Gamma(q+1)}w_1(\alpha)=\frac{(1-\alpha)\Gamma(q-\alpha)}{\Gamma(2-\alpha)\Gamma(q+1)}\frac{1}{n}\sum_{i=1}^n \lambda_i.
\label{sec1 : eq13}
\end{eqnarray}
The determinant of $\mathcal{A}_n$ is
\begin{eqnarray}
\textrm{det}(\mathcal{A}_n)=\left(\frac{K_{\alpha}}{2\left|\cos \left(\frac{\pi \alpha}{2}\right)\right|\epsilon^{\alpha}}\right)^n\prod_{i=1}^n \lambda_i>0
\label{sec1 : eq14}
\end{eqnarray}
since $\lambda_i> 0, \, \forall i=1,\cdots , n$. 
\item[($ii$)] For $\beta\in\{-1,1\}$ the operator $\hat{\mathcal{A}}(\alpha,\beta=\pm 1)$ is one-sided and the corresponding matrix is a sum of a lower (or upper)-diagonal matrix plus a matrix $\mathcal{J}_n$ with elements $\delta_{i,i+1}=2w_0$ (or $\delta_{i-1,i}=2w_0$). For $n=2$ we have two real eigenvalues while $\forall n\geq 3$ we have complex and real eigenvalues. 
\end{enumerate}
In the most general case ($\alpha\in (1,2)$ and $\beta \in (-1,1)$) the matrix $\mathcal{A}_{n}(\alpha,\beta)$ decomposes into a sum of a symmetric $\alpha$-dependent matrix and an antisymmetric $\beta$-dependent matrix. The eigenvalues in this case are both real and complex.
\section{Asymptotic behaviour of the product of eigenvalues for $\mathcal{A}_n$}
\label{sec2}
In this section we examine the validity of the conditions under which the Szeg\"{o}'s strong limit theorem holds \cite{Ref10} and determine the asymptotic behaviour of the product of eigenvalues. 
\begin{definition}
The set of all functions 
\begin{eqnarray}
f(z)=\sum_{k\in {\mathbb Z}}f_k z^k, \quad z=e^{i\theta}\in{\mathbb T}, \quad {\rm{with}} \quad \sum_{k\in {\mathbb Z}}|f_k|<\infty
\label{sec2 : eq1} 
\end{eqnarray}
is denoted by $W:=W({\mathbb T})$ and is called the Wiener algebra. 
\end{definition}
\begin{definition}
The set of all functions \rf{sec2 : eq1} which belong to $L^2$ and satisfy
\begin{eqnarray}
\sum_{k\in {\mathbb Z}}(|k|+1)|f_k|^2<\infty
\label{sec2 : eq2} 
\end{eqnarray}
is denoted by $B_2^{1/2}$ and is referred as a Besov space.
\end{definition}
\begin{proposition}
The symbol $f$ given by \rf{sec1 : eq17} belongs to the space $W\cap B^{1/2}_2$, it has a zero on ${\mathbb T}$ at $\theta=0$ and the winding number of f around the origin vanishes, $\textrm{wind}(f,0)=0$. 
\end{proposition}
\begin{proof} 
\end{proof}
$f$ is an element of the Wiener algebra since 
\begin{eqnarray}
\sum_{n\in {\mathbb Z}}|f_n|&=& \frac{K_{\alpha}}{2\left|\cos \left(\frac{\pi \alpha}{2}\right)\right|}\left(|f_0|+\sum_{n=1}^{\infty} (|f_{-n}|+|f_n|)\right) \non \\
&\leq& \frac{K_{\alpha}}{\left|\cos \left(\frac{\pi \alpha}{2}\right)\right|}\left(|w_1|-w_1+2\sum_{n=0}^{\infty}w_n\right) \leq \frac{3K_{\alpha}|\alpha|}{\left|\cos \left(\frac{\pi \alpha}{2}\right)\right|}<\infty.
\label{sec2 : eq3}
\end{eqnarray}
Also $f$ belongs to the Besov space since $f\in L^2$ and (see \rf{apA3} and \rf{apA5} for the proof)
\begin{eqnarray}
\sum_{n\in {\mathbb Z}}(|n|+1) |f_n|^2&\leq& \frac{K_{\alpha}^2}{4\cos^2 \left(\frac{\pi \alpha}{2}\right)} \left(4w_1^2 +16(|w_0|+|w_2|)^2+8\sum_{n=2}^{\infty}(n+1)w_{n+1}^2\right) \non \\
&=& \frac{4K_{\alpha}^2}{\cos^2 \left(\frac{\pi \alpha}{2}\right)} \left(\alpha\left(\frac{3}{4}\alpha-1\right)+1+\frac{16}{9\pi}\right)<\infty.
\label{sec2 : eq4}
\end{eqnarray}
The symbol $f$ has a zero on ${\mathbb T}$ at $\theta=0$. Transversing ${\mathbb T}$ once the argument of $f(e^{i\theta})$ will return to its original zero value and the winding number vanish.
\begin{theorem}[Szeg\"{o}'s strong limit theorem]
If $f\in W\cap B^{1/2}_2$ has no zeros on ${\mathbb T}$ and $\textrm{wind}(f,0)=0$ then
\begin{eqnarray}
\lim_{n\rightarrow \infty}\frac{\textrm{det}(\mathcal{A}_n (\alpha,\beta=0))}{G(f)^n}=E(f)
\label{sec2 : eq5}
\end{eqnarray}
where 
\begin{eqnarray}
G(f)&=& \exp(\log f)_0  \label{iu2}\\
E(f)&=& \exp\left(\sum_{k=1}^{\infty} k(\log f)_k (\log f)_{-k}\right) 
\label{sec2 : eq6}
\end{eqnarray}
and $(\log f)_k$ the Fourier coefficients of $\log (f(e^{i\theta}))$.
\end{theorem}
In our case all requirements are fullfiled except that $f$ has zeros at $\theta=2m\pi, \, m\in{\mathbb Z}$. Subtracting these points from ${\mathbb T}$ and applying the theorem for $\alpha \in (1,2)$ and $\beta \in [1,-1]$ the Fourier coefficients, using the cosine and sine series, are given by
\begin{eqnarray}
(\log f)_k&=&\frac{1}{2\pi}\int_{-\pi}^{\pi} \log(f(e^{i\theta}))e^{-ik\theta} d\theta \non \\
&=&\!\!\! \frac{1}{2\pi}\left[\int_{0}^{\pi}\log(u^2(\theta)+v^2(\theta)) \cos(k\theta) d\theta + i\int_{0}^{\pi}\log\left(\frac{u(\theta)+iv(\theta)}{u(\theta)-iv(\theta)}\right)\sin(k\theta)d\theta\right]\!\!.
\label{sec2 : eq7}
\end{eqnarray}
The calculation of such integrals is cumbersome and we investigate only the $S_{\alpha}(c)$ laws. Equation \rf{sec2 : eq7} is reduced to
\begin{eqnarray}
(\log f)_k=\frac{1}{\pi}\int_0^{\pi}\log (f(e^{i\theta}))\cos(k\theta) d\theta
\label{sec2 : eq8}
\end{eqnarray}
bearing in mind that $\log (f(e^{i\theta}))$ is an even function.
\begin{theorem}
The asymptotic behaviour of eigenvalues for the problem \rf{sec1 : eq10} is given by \rf{sec3 : eq16}. 
\end{theorem}
\begin{proof}
\end{proof}
The generating function 
\begin{eqnarray}
\log(f(e^{i\theta}))= \log K_{\alpha}-\log |\cos(\frac{\alpha \pi}{2})|+\log \left|2 \cos \frac{\theta}{2}\right|^{\alpha} 
+ \log \left(\cos\left[\theta \left(1-\frac{\alpha}{2}\right)\right]\right)
\label{sec3 : eq9}
\end{eqnarray}
is $T=2\pi n_1$ periodic (with $\theta\in(-\pi,\pi)$) and using \rf{sec2 : eq8} we study the following two cases
\begin{description}
\item[($i$)] k=0. The corresponding Fourier coefficient is given by
\begin{eqnarray}
(\log f)_0&=&\log K_{\alpha}-\log|\cos(\frac{\alpha \pi}{2})|+\frac{1}{\pi}\int_0^{\pi}\log \left|2 \cos \frac{\theta}{2}\right|^{\alpha} d\theta \non \\
&+&\frac{1}{\pi}\int_0^{\pi}\log \left[\cos \left(\theta(1-\frac{\alpha}{2})\right)\right]d\theta. 
\label{sec3 : eq10}
\end{eqnarray}
The first integral on the right handside vanishes since 
\begin{displaymath}
\int_0^{\pi}\log \left|2 \cos \frac{\theta}{2}\right|d\theta=\int_0^{\pi}\log \left|2 \sin \frac{\theta}{2}\right|d\theta=0. 
\end{displaymath}
The second one after the change of variable, 
\begin{displaymath}
\theta\left(1-\frac{\alpha}{2}\right)=u, \quad 
\end{displaymath}
becomes
\begin{eqnarray}
\frac{1}{\pi}\int_0^{\pi}\log \left[\cos \left(\theta(1-\frac{\alpha}{2})\right)\right]d\theta &=& \frac{2}{(2-\alpha)\pi}\int_0^{\pi\left(1-\frac{\alpha}{2}\right)} \log(\cos u)du \non \\
&=& -\frac{2}{(2-\alpha)\pi}L\left(\pi\left(1-\frac{\alpha}{2}\right)\right)
\label{sec3 : eq11}
\end{eqnarray}
where $L$ is Lobachevskiy's function \cite{Ref14}. Thus
\begin{eqnarray}
(\log f)_0=\log K_{\alpha}-\log|\cos(\frac{\alpha \pi}{2})| -\frac{2}{(2-\alpha)\pi}L\left(\pi\left(1-\frac{\alpha}{2}\right)\right).
\label{sec3 : eq12}
\end{eqnarray}
\item[($ii$)] $k \in {\mathbb Z_+}$. The Fourier coefficients are given by
\begin{eqnarray}
(\log f)_k &=& \frac{\alpha}{\pi}\int_0^{\pi}\log \left|2 \cos \frac{\theta}{2}\right|\cos( k\theta)d\theta \non \\
&+& \frac{1}{\pi}\int_0^{\pi}\log \left[\cos \left(\theta(1-\frac{\alpha}{2}) \right)\right]\cos( k\theta)d\theta \non \\
&=& I_1+I_2
\label{sec3 : eq13}
\end{eqnarray}
with 
\begin{eqnarray}
I_1=(-1)^{k+1}\frac{\alpha}{2k}
\label{sec3 : eq14}
\end{eqnarray}
and  
\begin{eqnarray}
I_2&=&\frac{2}{(2-\alpha)\pi}\int_0^{\pi\left(1-\frac{\alpha}{2}\right)}\log(\cos(u))\cos \left(\frac{2ku}{2-\alpha}\right)du \non \\
&=& \frac{1}{2k\pi}\int_0^{\pi\left(1-\frac{\alpha}{2}\right)}\left(\frac{\cos\left(\frac{k}{1-\frac{\alpha}{2}}-1\right)u-\cos\left(\frac{k}{1-\frac{\alpha}{2}}+1\right)u}{\cos u}\right)du. 
\label{sec3 : eq15}
\end{eqnarray}
\end{description}
The second equality in \rf{sec3 : eq15} is justified by partial integration.
Thus \rf{sec3 : eq13} is written as 
\begin{eqnarray}
(\log f)_k =(-1)^{k+1}\frac{\alpha}{2k}+\frac{1}{2k\pi}\int_0^{\pi\left(1-\frac{\alpha}{2}\right)}\left(\frac{\cos\left(\frac{k}{1-\frac{\alpha}{2}}-1\right)u-\cos\left(\frac{k}{1-\frac{\alpha}{2}}+1\right)u}{\cos u}\right)du.
\label{sec3 : eq16}
\end{eqnarray}
 
We study two physically well-known  cases: 
\begin{description}
\item[($1$)] The Gaussian law with $\alpha=2$ and $\beta=0$. Substituting $f(e^{i\theta})=K_{2}\left(2\sin \frac{\theta}{2}\right)^2$ into \rf{sec2 : eq8} we obtain \cite{Ref12}:
\begin{description}
\item[($i$)] $k=0$
\begin{eqnarray} 
(\log f)_0=\log K_{2}+\frac{2}{\pi}\int_0^{\pi} \log \left(2\sin\left(\frac{\theta}{2}\right)\right) d\theta=\log K_{2}.
\label{sec3 : eq1}
\end{eqnarray} 
\item[($ii$)] $k\neq 0$ and $\theta\neq 2m\pi, \, m\in {\mathbb Z}$
\begin{eqnarray}
(\log f)_k=\frac{2}{\pi}\int_0^{\pi} \log \left(2\sin\left(\frac{\theta}{2}\right)\right) \cos(k\theta)d\theta=-\frac{1}{k}, \quad k\in{\mathbb N}.
\label{sec3 : eq2} 
\end{eqnarray}
\end{description}
Applying theorem \rf{sec2 : eq5} we get
\begin{eqnarray}
\lim_{n\rightarrow \infty}\frac{\textrm{det}(\mathcal{A}_n (\alpha,\beta=0))}{K_2^n}=e^{\sum_{k=1}^{\infty}\frac{1}{k}}=e^{C}\lim_{m\rightarrow \infty}m
\label{sec3 : eq3}
\end{eqnarray}
where the Euler's constant $C$ is defined by $C=\lim_{m\rightarrow \infty}\left(\sum_{k=1}^m \frac{1}{k}-\log m\right)$. Expressions \rf{sec3 : eq3} and \rf{sec01 : eq91} differ only by the constant factor $e^{C}$. This result is also justified by (39) of \cite{Ref13} in the limit of the unit circle $R=1$ and setting $a=0$.
\item[($2$)] The Holdsmark law with $\alpha=\frac32$ and $\beta=0$ \cite{Ref15, Ref16}. The Fourier coefficients from \rf{sec3 : eq12} and \rf{sec3 : eq16} are found to be 
\begin{eqnarray}
(\log f(\alpha=\frac{3}{2}))_0 &=& \log K_{\frac{3}{2}}-\log(\sqrt{2})+\frac{2}{\pi}G  \label{iu15}\\ 
(\log f(\alpha=\frac{3}{2}))_k &=& (-1)^{k+1}\frac{3}{4k}-\frac{1}{\pi k}\left(\frac{\pi}{4}-\sum_{m=1}^{2k}\frac{\sin \left(\frac{m\pi}{2}\right)}{m}\right)\non \\
&=& (-1)^{k}\left(-\frac{3}{4k}+\frac{1}{4k}\left(\Psi(\frac{2k+1}{4})-\Psi(\frac{2k-1}{4})\right)\right)
\label{sec3 : eq17}
\end{eqnarray}
where $G$ is Catalan's constant defined by 
\begin{displaymath}
G=\sum_{m=0}^{\infty}\frac{(-1)^m}{(2m+1)^2}=0.915965\cdots
\end{displaymath}
and the $\Psi$ function is defined by $\Psi(x)=\frac{d \log\Gamma(x)}{dx}$.
\end{description}
When $\alpha$ takes values on the Farey series \footnote{The Farey series $\mathcal{F}_n$ of order n is the
ascending series of irreducible fractions between $0$ and $1$ whose
denominators do not exceed $n$. Thus $\alpha=\frac{p}{q}$ belongs in
$\mathcal{F}_n$ if
\begin{displaymath}0\leq p \leq q \leq n ,
\quad (p,q)=1 \end{displaymath}
where $(,)$ denotes the highest common divisor of two integers.} then the zero and nonzero Fourier coefficients are expressed in terms of polylogarithmic and hypergeometric functions. 

\section{Conclusions}  

At quantum level using the infinitesimal generator of time translations with vanishing skewness  and the definition of the Gr\"{u}nwald-Letnikov fractional derivative for the one-dimensional infinite well potential with Dirichlet boundary conditions, we established the correspondence between the disctretized version of the fractional Schr\"{o}dinger problem and Toeplitz matrices. The next step was to check the conditions under which the Sz\"{e}go's strong limit theorem is valid and subsequently to determine the asymptotic behaviour of the product of eigenvalues given by \rf{sec2 : eq6} and \rf{sec3 : eq16}. An open question which deserves investigation is the general $S_{\alpha}(\beta,c,\tau)$ law. 
\par Finally, from the physical point of view there is a plethora of applications related to the fractional diffusion equation \cite{Ref8} but not to the Schr\"{o}dinger equation. Recently, the authors of \cite{Ref17} constructed a one-dimensional lattice model with a hopping particle and numerically obtained the eigenvalues and eigenfunctions in a bounded domain with different boundary conditions. It would be interesting to study the asymptotic behaviour of the product of eigenvalues in such a model using the Gr\"{u}nwald-Letnikov fractional derivative. 

\addcontentsline{toc}{subsection}{Appendix A }
\section*{Appendix A }
\label{apA}
\renewcommand{\theequation}{A.\arabic{equation}}
\setcounter{equation}{0}
\textbf{Proof of \rf{sec1 : eq6}}\\
The infinite series \rf{sec1 : eq6} is written as 
\beqr
\sum_{n=0}^{\infty}(-1)^n \left(\begin{array}{l} \alpha \\ n \end{array} \right)&=& \frac{\alpha(\alpha-1)\Gamma(-\alpha)}{\Gamma(2-\alpha)}\sum_{n=0}^{\infty}(-1)^{2n}(-\alpha)_n \frac{1^n}{\Gamma(n+1)}\non \\
&=& {}_1 F_0(-\alpha;1)=0, \quad \alpha\in(0,2)/\{1\}
\label{apAb1}
\feqr
and converges to zero since ${}_1 F_0(-\alpha;z)=(1-z)^{\alpha}$.\\
 
\textbf{Proof of \rf{sec1 : eq17a1}}\\
We prove first the identity 
\beqr
&& \sum_{n=1}^{\infty} \cos(n\theta)\frac{\Gamma(n+1-\alpha)}{\Gamma(n+2)}=\frac{1}{4}\Gamma(2-\alpha)\left(e^{i\theta}{}_2 F_1(1,2-\alpha;3;e^{i\theta})+e^{-i\theta}{}_2 F_1(1,2-\alpha;3;e^{-i\theta})\right)\non\\
&=& \frac{1}{2}\Gamma(-\alpha)\left[2^{\alpha}\left(\sin \frac{\theta}{2}\right)^{\alpha}\cos \frac{1}{2}\left(\theta(2-\alpha)+\pi \alpha\right)+2(\alpha-\cos \theta)\right].
\label{apA1}
\feqr
Writing $\cos(\theta)=\frac{1}{2}(e^{i\theta}+e^{-i\theta})$ it is enough to calculate
\beqr
\sum_{n=1}^{\infty}\frac{\Gamma(n+1-\alpha)}{\Gamma(n+2)}e^{in\theta}&=& e^{i\theta}\frac{\Gamma(2-\alpha)}{\Gamma(3)}\sum_{n=1}^{\infty}\frac{(1)_n (2-\alpha)_n}{(3)_n} \frac{e^{in\theta}}{n!} \non \\
&=& e^{i\theta}\frac{\Gamma(2-\alpha)}{\Gamma(3)}{}_2 F_1(1,2-\alpha;3;e^{i\theta}) \non \\
&=& e^{-i\theta} \Gamma(-\alpha)\left[(1-e^{i\theta})^{\alpha}-1+\alpha e^{i\theta} \right]\non \\
&=& e^{-i\theta} \Gamma(-\alpha)\left[{}_1 F_0(-\alpha;e^{i\theta})-1+\alpha e^{i\theta})\right]
\label{apA2}
\feqr
where $(a)_n$ is the Pochhammer's symbol
\begin{displaymath} 
(a)_n\equiv a(a-1)\cdots(a+n-1)=\frac{\Gamma(a+n)}{\Gamma(a)}, \quad (a)_0\equiv 1, \quad (1)_n=n! 
\end{displaymath}
The second term of the sum \rf{apA1} is derived from \rf{apA2} by complex conjugation. \\

\textbf{Proof of \rf{sec2 : eq4}}\\
In \rf{sec2 : eq4} we used that the sum 
\beqr
\sum_{n=2}^{\infty} a_n =\sum_{n=2}^{\infty}n \left(\frac{\Gamma(n+1-\alpha)}{\Gamma(n+2)}\right)^2<\infty
\label{apA3}
\feqr
converges. It is easily checked using D'Alembert's test
\beqr
\frac{a_{n+1}}{a_n}=\frac{(n+1)}{n}\frac{(n+1-\alpha)}{(n+2)^2}<1, \quad \forall n\geq 2
\label{apA4}
\feqr
and the positivity of $a_n$'s.
Also we used the identity
\beqr
\sum_{n=2}^{\infty}\left(\frac{\Gamma(n+1-\alpha)}{\Gamma(n+2)}\right)^2=\left(\frac{\Gamma(3-\alpha)}{\Gamma(4)}\right)^2 {}_3F_2(1,3-\alpha,3-\alpha;4,4;1).
\label{apA5}
\feqr
\section*{Appendix B }
\label{apB}
\renewcommand{\theequation}{B.\arabic{equation}}
\setcounter{equation}{0}
Let 
\begin{eqnarray}
J_{ij}:=\delta_{i,n+1-j}, \,\, i,j=1,2,\cdots,n
\label{apB1}
\end{eqnarray}
be the matrix with ones along the secondary diagonal and zeroes elsewhere. For a real, symmetric and Toeplitz matrix one can prove that \cite{Ref11}
\begin{eqnarray}
[J,\mathcal{A}_n]=0
\label{apB2}
\end{eqnarray} 
Using \rf{apB2} and since $J^2=I$ this implies that $\mathcal{A}_n x=\lambda x$ if and only if $\mathcal{A}_n(Jx)=\lambda (Jx)$. If $\lambda$ has multiplicity one then from $\parallel x\parallel_2=\parallel Jx\parallel_2$ and $Jx=cx$ we conclude that $c=\pm 1$. 
\begin{definition}
An $n$-dimensional vector $x$ will be called symmetric if
\begin{eqnarray}
Jx=x
\label{apB3} 
\end{eqnarray}
or skew-symmetric if 
\begin{eqnarray}
Jx=-x.
\label{apB4}
\end{eqnarray}
 \end{definition}
\begin{definition}
The eigenvalue $\lambda$ of $\mathcal{A}_n$ is even (or odd) if it has an associated symmetric (or skew-symmetric) eigenvector $x$. 
\end{definition}

\bibliographystyle{plain}

\end{document}